\documentclass{eptcs-modified}

\usepackage{amsmath}
\usepackage{amssymb}
\usepackage{url}
\usepackage{amsthm}
\usepackage{graphicx}

\theoremstyle{definition}
\newtheorem{theorem}{Theorem}
\newtheorem{definition}[theorem]{Definition}

\newtheorem{corollary}[theorem]{Corollary}
\newtheorem{proposition}[theorem]{Proposition}
\newtheorem{example}[theorem]{Example}
\newtheorem{lemma}[theorem]{Lemma}
\newtheorem{observation}[theorem]{Observation}

\newtheorem{remark}[theorem]{Remark}

\newcommand{\name}[1]{\textsc{#1}}
\newcommand{\reg}{\mathbf{Reg}}
\newcommand{\mon}{\mathbf{Mon}}
\newcommand{\Baire}{{\mathbb{N}^\mathbb{N}}}
\newcommand{\dom}{\textrm{dom}}
\newcommand{\id}{\textrm{id}}
\newcommand{\Sierp}{Sierpi\'nski }

\newcommand{\hide}[1]{}

\newcommand{\omitted}[1]{#1}

\newcommand{\lit}[1]{\texttt{\textup{#1}}}

\newcommand{\Nset}{\mathbb N}

\newcommand{\Bset}{\mathbf{Reg}}

\begin{document}

% first the title is needed
\title{Function spaces for second-order polynomial time}

\author{
Akitoshi Kawamura
\institute{Department of Computer Science\\ University of Tokyo, Japan}
\email{kawamura@is.s.u-tokyo.ac.jp}
\and
Arno Pauly
\institute{Computer Laboratory\\ University of Cambridge, United Kingdom}
\email{Arno.Pauly@cl.cam.ac.uk}
}

\def\titlerunning{Function spaces for second-order polynomial time}
\def\authorrunning{A.~Kawamura \& A.~Pauly}

\maketitle

\begin{abstract}
In the context of second-order polynomial-time computability, we prove that there is no general function space construction. We proceed to identify restrictions on the domain or the codomain that do provide a function space with polynomial-time function evaluation containing all polynomial-time computable functions of that type.

As side results we show that a polynomial-time counterpart to admissibility of a representation is not a suitable criterion for natural representations, and that the Weihrauch degrees embed into the polynomial-time Weihrauch degrees.
\end{abstract}

%\keywords{cartesian closed, computational complexity, higher order, computable analysis, admissible representation, Weihrauch reducibility}

\section{Introduction}
Computable analysis (e.g.~\cite{weihrauchd}) deals with computability questions for operators from analysis such as integration, differentiation, Fourier transformation, etc.. In general, the actual computation is envisioned to be performed on infinite sequences over some finite or countable alphabet, this model is then lifted to the spaces of interest by means of representations. Thus, an adequate choice of representations for the various relevant spaces is the crucial foundation for any investigation in computable analysis.

At first, the search for good representations proceeded in a very ad-hoc fashion, exemplified by \name{Turing}'s original definition of a computable real number as one with computable decimal expansion \cite{turing} and later correction to one with a computable sequence of nested rational intervals collapsing to the number \cite{turingb}\footnote{This choice of a representation, which is indeed a \emph{correct} one, is credited to \name{Brouwer} by \name{Turing}.}.

The development of more systematic techniques to identify good representations had two interlocked main components: One, the identification of \emph{admissibility} as the central criterion whenever the space in question already carries a natural topology by \name{Kreitz} and \name{Weihrauch} \cite{kreitz} and later \name{Schr\"oder} \cite{schroder}. Two, the observation that one can form function spaces in the category of represented spaces (e.g.~\cite{weihrauchk}, \cite{bauer3}). Using the ideas of synthetic topology \cite{escardo}, this suffices to obtain good representations of spaces just from their basic structure\footnote{The concept of structure here goes beyond topologies, as witnessed e.g.~by the treatment of hyperspaces of measurable sets and functions in \cite{pauly-descriptive,pauly-descriptive-lics} or of the countable ordinals in \cite{pauly-ordinals-mfcs,zhenhao}.} (demonstrated in \cite{pauly-synthetic}).

While computable analysis has obtained a plethora of results, for a long time the aspect of computational complexity has largely been confined to restricted settings (e.g.~\cite{weihrauchf}) or non-uniform results (e.g.~\cite{ko}). This was due to the absence of a sufficiently general theory of second-order polynomial-time computability --  a gap which was filled by \name{Cook} and the first author in \cite{kawamura}. This theory can be considered as a refinement of the computability theory. In particular, this means that for doing complexity theory, one has to choose well-behaved representations for polynomial-time computation out of the equivalence classes w.r.t.~computable translations.

Various results on individual operators have been obtained in this new framework \cite{kawamura5,kawamura4,kawamura3,roesnick}, leaving the field at a very similar state as the early investigation of computability in analysis: While some indicators are available what good choices of representations are, an overall theory of representations for computational complexity is missing. Our goal here is to provide the first steps towards such a theory by investigating the role of admissibility and the presence of function spaces for polynomial-time computability.

\section{Background on second-order polynomial-time computability}

We will use (a certain class of) string functions to encode the objects of interest. We fix some alphabet $\varSigma$.
We say that a (total) function $\varphi \colon \varSigma ^* \to \varSigma ^*$ is
\emph{regular}
if it preserves relative lengths of strings in the sense that
$\lvert \varphi (u) \rvert \leq \lvert \varphi (v) \rvert$
whenever $\lvert u \rvert \leq \lvert v \rvert$.
We write $\Bset$ for the set of all regular functions.
We restrict attention to regular functions (rather than using
all functions from $\varSigma ^*$ to $\varSigma ^*$)
to keep the notion of their \emph{size} (to be defined shortly) simple.

We use an oracle Turing machine (henceforth just ``machine'')
to convert regular functions to regular functions (Figure~\ref{figure: oracle machine}).
\begin{figure}
\begin{center}
\includegraphics[scale=1.1]{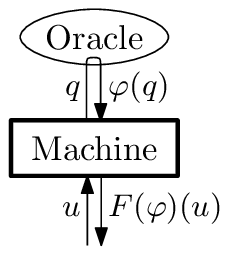}
\caption{A machine computing a function $F \colon \Bset \to \Bset$.}
\label{figure: oracle machine}
\end{center}
\end{figure}

\begin{definition}
\label{definition: oracle machine}
A machine~$M$
computes a partial function $F : \subseteq \Bset \to \Bset$
if
for any $\varphi \in \dom F$,
the machine~$M$ on oracle~$\varphi$ and any string~$u$ outputs $F(\varphi)(u)$ and halts.
\end{definition}

\begin{remark}
For computability,
this is equivalent to the model where
a Turing machine converts infinite strings to infinite strings.
For the discussion of polynomial-time computability, however,
we really need to use strings functions
in order to encode information efficiently
and to measure the input size,
as we will see below.
\end{remark}

Regular functions map strings of equal length to strings of equal length.
Therefore it makes sense to
define the
\emph{size}
$\lvert \varphi \rvert \colon \Nset \to \Nset$
of a regular function~$\varphi$
to be the (non-decreasing) function
$\lvert \varphi \rvert (\lvert u \rvert) = \lvert \varphi (u) \rvert$.
We will use $\mon$ to denote the strictly monotone functions from $\Nset$ to $\Nset$.
For technical reasons, we will tacitly restrict ourselves to those regular functions $\varphi$ with $|\varphi| \in \mon$, this does not impede generality\footnote{Given some $\varphi \in \reg$, let $\varphi'$ be defined by $\varphi'(v) = v\varphi(v)$. Then the function $\mathord\cdot' : \reg \to \reg$ is polynomial-time computable, and has a polynomial-time computable inverse. Moreover, $|\varphi'| \in \mon$ for all $\varphi \in \reg$.}.

We will make use of a polynomial-time computable pairing function $\langle, \rangle : \varSigma^* \times \varSigma^* \to \varSigma^*$, which we want\footnote{While this choice is a bit wasteful, it is useful for technical reasons, and ultimately does not matter for polynomial-time computability.} to satisfy $|\langle u, v\rangle| = |u| \times |v|$. This is then lifted to a pairing function on $\reg$ via $\langle \varphi, \phi \rangle(u) = \langle \varphi(u), \psi(u)\rangle$, and to a mixed pairing function for $\langle -,-\rangle : \varSigma^* \times \reg \to \reg$.

Now we want to define what it means for a machine to run in polynomial time.
Since $\lvert \varphi \rvert$ is a function,
we begin by defining polynomials in a function,
following the idea of Kapron and Cook~%
\cite{kapron2}.
\emph{Second-order polynomials}
(in type-$1$ variable~$\lit L$ and type-$0$ variable~$\lit n$)
are defined
inductively as follows:
a positive integer is a second-order polynomial;
the variable~$\lit n$ is also a second-order polynomial;
if $P$ and $Q$ are second-order polynomials,
then so are $P + Q$, $P \cdot Q$ and $\lit L (P)$.
An example is
\begin{equation}
 \label{equation: second-order polynomial example}
  \lit L \bigl( \lit L (\lit n \cdot \lit n) \bigr) + \lit L \bigl( \lit L (\lit n) \cdot \lit L (\lit n) \bigr) + \lit L (\lit n) + 4.
\end{equation}
A second-order polynomial~$P$ specifies a
function, which we also denote by $P$,
that takes functions $L \in \mon$
to another function $P (L) \in \mon$
in the obvious way.
For example,
if $P$ is the above second-order polynomial~\eqref{equation: second-order polynomial example}
and $L (n) = n ^2$,
then $P (L)$ is given by
\begin{equation}
  P (L) (n) = \bigl( (n \cdot n) ^2 \bigr) ^2 + (n ^2 \cdot n ^2) ^2 + n ^2 + 4
            = 2 \cdot n ^8 + n ^2 + 4.
\end{equation}
As in this example,
$P (L)$ is a (usual first-order) polynomial if $L$ is.

\begin{definition}
 \label{definition: bounded by second-order polynomial}
A machine~$M$
runs in \emph{polynomial time}
if there is a second-order polynomial~$P$ such that,
given any $\varphi \in \Bset$ as oracle
and any $u \in \varSigma ^*$ as input,
$M$ halts within $P (\lvert \varphi \rvert) (\lvert u \rvert)$ steps.
\end{definition}

This defines the class of
(polynomial-time) computable
functions from $\Bset$ to $\Bset$.
We can suitably define
some other complexity classes related to
nondeterminism or space complexity,
as well as the notions of reduction and hardness~%
\cite{kawamura}.

A
\emph{representation}
$\delta$ of a set~$X$ is formally
a partial function from $\Bset$ to $X$ that is
surjective---that is,
for each $x \in X$,
there is at least one $\varphi \in \reg$
with $\delta (\varphi) = x$.
We say that $\varphi$ is a
\emph{$\delta$-name}
of $x$. A \emph{represented space} is a pair $\mathbf{X} = (X, \delta_X)$ of a set $X$ together with a representation $\delta_X$ of it. For a function $f : \subseteq \mathbf{X} \to \mathbf{Y}$ between represented spaces $\mathbf{X}$, $\mathbf{Y}$ and $F : \subseteq \reg \to \reg$, we call $F$ a realizer of $f$ (notation $F \vdash f$), iff $\delta_Y(F(p)) = f(\delta_X(p))$ for all $p \in \dom(f\delta_X)$. A map between represented spaces is called (polynomial-time) computable, iff it has a (polynomial-time) computable realizer.

Type-2 complexity theory generalizes classical complexity theory, as we can regard the objects of the latter as special $\reg$-represented spaces. In the following, we will in particular understand $\mathbb{N}$ to be represented via $\delta_\mathbb{N}(\varphi) = |\varphi(0)|$, i.e.~using an adaption of the unary representation (although not much would change if the binary representation were used instead).

\section{Some properties of second-order polynomials}
We will establish some properties of second-order polynomials as the foundation for our further investigations. Our primary interest is in capturing the rates of asymptotic growth in both arguments, or, rather, a generalization of the notion of asymptotic growth of first-order polynomials (and functions in general) to second-order polynomials (and functionals in general). We arrive at the following definition:

\begin{definition}
\label{def:asymptotic}
Let $P$, $Q$ be second-order polynomials. We write $P \in \mathcal{O}^2(Q)$ iff
\[\exists q \in \mon, k \in \mathbb{N} \quad \forall p \in \mon, n \in \mathbb{N} \quad P(p)(n) \leq Q(p \times q)((n+1)^k)\]
\end{definition}

We subsequently introduce the notion of the second-order degree of a second-order polynomial -- just as the first-order degree is intricately to asymptotic growth of first-order polynomials, the second-order degree will prove to be a valuable tool in the classifications required for our work.

\begin{definition}[\footnote{We point out that this definition differs from the one given in previous versions, in particular in \cite{pauly-kawamura}.}]
\label{def:degree}
The second-order degree of a second-order polynomial (denoted by $\deg$) will be defined interleaved with its type $(\operatorname{type}$), which is only used for the definition here:
\begin{itemize}
\item $\deg(1) = 0$, $\operatorname{type}(1) = \mathbf{m}$
\item $\deg(n) = 0$, $\operatorname{type}(n) = \mathbf{m}$
\item $\deg(P+Q) = \max \{\deg(P), \deg(Q)\}$, if $\operatorname{type}(P) = \mathbf{a} \wedge \deg(P) = \max \{\deg(P), \deg(Q)\}$ or $\operatorname{type}(Q) = \mathbf{a} \wedge \deg(Q) = \max \{\deg(P), \deg(Q)\}$, then $\operatorname{type}(P + Q) = \mathbf{a}$, else $\operatorname{type}(P + Q) = \mathbf{m}$.
\item $\deg(\lit L(P)) = \deg(P) + 1$, $\operatorname{type}(\lit L(P)) = \mathbf{a}$
\item $\deg(P \times Q) = \max \{\deg(P), \deg(Q)\} + 1$, if $\operatorname{type}(P) = \mathbf{a} \wedge \deg(P) = \max \{\deg(P), \deg(Q)\}$ or $\operatorname{type}(Q) = \mathbf{a} \wedge \deg(Q) = \max \{\deg(P), \deg(Q)\}$
\item $\deg(P \times Q) = \max \{\deg(P), \deg(Q)\}$ else
\item $\operatorname{type}(P \times Q) = \mathbf{m}$
\end{itemize}
\end{definition}

Informally, the degree counts the number of nested function applications plus the number of type 1 polynomials of degree greater than 1 applied intermittently. A related notion is the \emph{depth} of a second-order polynomial introduced as a measure of complexity by \name{Kapron} and \name{Cook} \cite{kapron2}. The depth simply counts the number of nested function applications, we thus find that $\operatorname{depth}(P) \leq \deg P \leq \operatorname{depth}(P)$ for all second-order polynomials $P$ (and for fixed depth, the degree can vary over the entire interval given). A further related concept -- the hyperdegree -- was recently suggested by \name{Ziegler} \cite{ziegler-dagstuhl}, this is a first-order polynomial describing, in some sense, the rate of growth of the second-order polynomial. The precise relationship between the hyperdegree and the second-order degree is currently unknown.

\begin{example}
Some examples of second order degrees:
\begin{itemize}
\item $\deg(\lit L (n^2)) = 1$
\item $\deg(\lit L (2 \cdot \lit L)) = 2$
\item $\deg(\lit L ((\lit L)^2)) = 3$
\item $\deg(\lit L ((\lit L)^2) + (\lit L(\lit L))^{1000}) = 3$
\end{itemize}
\end{example}

\begin{lemma}
\label{lem:structuretypem}
Let $Q$ be a second-order polynomial of type $\mathbf{m}$ with $\deg(Q) > 0$. Then there is a first-order multivariate polynomial $t$ and a finite number of second-order polynomials $Q_1, \ldots, Q_2$ with $Q = t(\lit L(Q_1), \ldots, \lit L(Q_n))$ and $\max_{i \leq n} \deg(Q_i) + 2 = \deg(Q)$.
\begin{proof}
We consider the term-tree of $Q$, and more specifically, all outer-most occurrences of $\lit L$. The subtrees below these induce the second-order polynomials $Q_i$. By replacing each occurrence of $\lit L$ and subsequent subtree by a different (first-order) variable $x_i$, we obtain a term-tree for the first-order polynomial $q$. Computing the degree of $Q$ inductively following Definition \ref{def:degree} yields the relationship $\max_{i \leq n} \deg(Q_i) + 2 = \deg(Q)$.
\end{proof}
\end{lemma}

Just as the degree of an ordinary polynomial uniquely determines its $\mathcal{O}$-notation equivalence class, we find a similar result for the second-order degree and second-order polynomials. The role of the monomials $x^n$ are taken by the second-order polynomials $P_n$ defined via $P_0(p)(k) = k$ and $P_{n+1}(p)(k) = p(P_n(p)(k))$.

\begin{lemma}
\label{lemma:degreeuniversal}
$Q \in \mathcal{O}^2(P_{\max \{\deg{Q},1\}})$ for any second-order polynomial $Q$.
\end{lemma}
\begin{proof}
By Definition \ref{def:asymptotic}, we need to show that for any second-order polynomial $Q$ there are $q \in \mon$ and $n \in \mathbb{N}$ such that $Q(p)(k) \leq P_{\max\{\deg(Q),1\}}(p \times q)((k+1)^n)$ for all $p \in \mon$, $k \in \mathbb{N}$. Our proofs proceeds by induction of the degree and the type, implicitly invoking Lemma \ref{lem:structuretypem} to ensure that our cases are indeed exhaustive.
\renewcommand\descriptionlabel[1]{\normalfont [\textbf{Case: }#1]}
\begin{description}
\item[$\deg(Q) = 0$] In this case $Q$ does not contain the first-order variable, and thus $Q (p)$ is an ordinary polynomial $q$.  We find $Q(p) (k) = q (k) = P _1 (q) (k) \leq P_1(\langle p, q\rangle)(k+1)$.

\item[$Q = t(Q_1,\ldots,Q_l)$, $\forall i \leq l \ . \ 0 < \deg(Q_i) < \deg(Q)$] By induction hypothesis, let $q_i$, $n_i$ be such that $Q_i(p)(k) \leq P_{\deg(Q) - 1}(p \times q_i)((k+1)^{n_i})$ for all $p \in \mon$, $k \in \mathbb{N}$. Let $q' := \langle q_1, \ldots, q_n\rangle$ and $n := \max_{i \leq l} n_i$. We now find that $\max_{i \leq l} Q_i(p)(k) \leq P_{\deg(Q) - 1}(p \times q')((k+1)^{n})$ for all $p \in \mon$, $k \in \mathbb{N}$.

    Next, let $t'$ be the univariate first-order polynomial obtained from $t$ by identifying all variables. We can now calculate:

    \begin{align*}
    Q (p)(k)  \leq & t'(\max_{i \leq l} Q_i(p)(k))\\
    & \leq  t'(P_{\deg(Q) - 1}(p \times q')((k+1)^{n}))\\
    & \leq (p \times q' \times t')(P_{\deg(Q) - 1}(p \times q' \times t')((k+1)^{n}))\\
    & = P_{\deg(Q)}(p \times q' \times t')((k+1)^{n})
    \end{align*}

    Thus, $q' \times t'$ and $n$ witness the claim.

\item[$Q = \lit L(Q_1)$, $\deg(Q_1) = 0$]
As pointed out above, $Q _1 (p)$ is some ordinary polynomial $q _1$ not dependent on $p$.
In particular, there is some $n \in \mathbb{N}$ such that
$q _1 (k) \leq (k + 1) ^n$.
We now find:
\begin{align*}
 Q (p) (k)
 & =
 p (Q _1 (p) (k)) \\
& =
 p (q _1 (k))\\
& \leq
 p \bigl( (k+1)^n \bigr)\\
& =
 P _1 (p) \bigl( (k + 1) ^n \bigr)\\
& \leq
 P _1 ( p \times 1 ) \bigl( (k + 1) ^n \bigr)\\
\end{align*}
\item[$Q = \lit L(Q_1)$, $\deg(Q_1) > 0$]
If $
 Q _1 (p) (k)
\leq
 P _{\deg (Q _1)} (p \times q ), (k + 1) ^n)
$, then:
\begin{align*}
 Q (p) (k)
& =
 p (Q _1 (p) (k))\\
& \leq
 p (P _{\deg (Q _1)} (p \times q) ((k + 1) ^n))\\
& \leq
 (p \times q ) (P _{\deg (Q_1)} (p \times q ) ((k+1)^n))\\
& =
 P _{\deg (Q)} (p \times q ) ((k + 1) ^n)\\
 \end{align*}
So the same witnesses working for $Q_1$ also work for $Q$.
%\item[$Q = Q_1 \times Q_2$, otherwise]
\end{description}
\end{proof}

\begin{lemma}
\label{lem:lgrowth}
Let $P$, $Q$ be second-order polynomials, $q \in \mon$ and $k \in \mathbb{N}$. If there are $p \in \mon$, $n \in \mathbb{N}$ such that $P(p)(n) > Q(p \times q)((n+1)^k)$, then for every $C \in \mathbb{N}$ there is a $p' \in \mon$ such that: \[(\lit L(P))(p')(n) > C + (\lit L(Q))(p' \times q)((n+1)^k)\]
\begin{proof}
By monotonicity and continuity of second-order polynomials, the premise depends only on the values of $p$ at $i < N := P(p)(n)$. We will obtain $p'$ by choosing $p'(N)$ sufficiently large, extending with $p'(N + i) = p'(N) + i$, and retaining $p'(i) = p(i)$ for $i < N$. By writing our the desired inequality, we find the criterion:
\[p'(N) > C + (p \times q)\left ( Q(p \times q)((n+1)^k) \right )\]
\end{proof}
\end{lemma}

\begin{corollary}
\label{corr:lgrowth}
If $\lit L(P) \in \mathcal{O}^2(\lit L(Q))$, then $P \in \mathcal{O}^2(Q)$.
\begin{proof}
We can weaken the claim of Lemma \ref{lem:lgrowth} for $C = 0$ by moving the universal quantifiers over $q$ and $k$ into the premise and conclusion. We arrive at the contraposition of the present statement.
\end{proof}
\end{corollary}

\begin{lemma}
\label{lem:pgrowth}
Let $P$, $Q$ be second-order polynomials, $q \in \mon$, $k \in \mathbb{N}$ and $r$ be a first-order polynomial with $\deg r \geq 2$. If there are $p \in \mon$, $n \in \mathbb{N}$ such that $P(p)(n) > Q(p \times q)((n+1)^k)$, then there is a $p' \in \mon$ such that:
\[r\left ((\lit L(P))(p')(n) \right ) > (\lit L(\lit L(Q)))(p' \times q)((n+1)^k)\]
\begin{proof}
For each $C \in \mathbb{N}$, we apply Lemma \ref{lem:lgrowth} to obtain some $p'_C$ with:
\[(\lit L(P))(p'_C)(n) > C + (\lit L(Q))(p'_C \times q)((n+1)^k)\]
Abbreviate $N := P(p)(n)$ and $M := \max \{0, ((\lit L(Q))(p' \times q)((n+1)^k)) - N\}$. By the actual construction used in the proof of Lemma \ref{lem:lgrowth}, we find that for $p'_C(N + M) = p'_C(N) + M$. Thus:
\[(\lit L(\lit L(Q)))(p' \times q)((n+1)^k) \leq p'_C(N) + M + q(N + M)\]
The desired inequality now is:
\[r(p'_C(N)) > p'_C(N) + M + q(N + M)\]
As $C$ goes to infinity, also $p'_C(N)$ goes to infinity. The other components remain unchanged. As by assumption $\deg r \geq 2$, the left-hand side will increase at least quadratically and the right-hand side only linear. Thus, by choosing $C$ sufficiently large, the inequality will become true.
\end{proof}
\end{lemma}

\begin{corollary}
\label{corr:pgrowth}
If $r(\lit L(P)) \in \mathcal{O}^2(\lit L(\lit L(Q)))$ with $\deg r \geq 2$, then $P \in \mathcal{O}^2(Q)$.
\begin{proof}
We can weaken the claim of Lemma \ref{lem:pgrowth} by moving the universal quantifiers over $q$ and $k$ into the premise and conclusion. We arrive at the contraposition of the present statement.
\end{proof}
\end{corollary}

\begin{theorem}
\label{theo:growth}
For $n \geq 1$ and a second-order polynomial $Q$ we find that $Q \in \mathcal{O}^2(P_n)$ iff $\deg Q \leq n$.
\begin{proof}
One direction is provided by Lemma \ref{lemma:degreeuniversal}. For the other direction, we use induction over the structure of $Q$ as provided by Lemma \ref{lem:structuretypem}, and use Corollaries \ref{corr:lgrowth}, \ref{corr:pgrowth} for the individual steps.
\end{proof}
\end{theorem}

\hide{
\begin{lemma}
\label{lemma:degreestrict}
For no $q \in \mon$, $n, m \in \mathbb{N}$ we have $P _{n + 1} (p) (k) \leq P _n (p \times q) ((k+1)^m)$ for all $p \in \mon$, $k \in \mathbb{N}$.
\end{lemma}
\begin{proof}
\omitted{Assume the contrary, and let $q$, $n$, $m$ witness the claim. We proceed to construct a $p$ such that $P_{n+1}(p,1) = p(P_n(p, 1)) > P_n(p \times q, 2^m)$, thus obtaining a contradiction. If $n = 0$, then choosing $p(k) = (k+1)^m + 1$ suffices. Next we define a $p$ in stages which will work for all remaining $n > 0$. For $1 \leq i \leq 2^m$, let $p(i) = 2^m + 1$. This implies $P_2(p, 1) = p(2^m + 1)$ and $P_1(p \times q, 2^m) = (2^m + 1) \times q(2^m)$. So choosing $p(2^m + 1) = ((2^m + 1) \times q(2^m)) + 1$ works. More generally, for $2^m + 1 \leq i \leq (2^m + 1) \times q(2^m)$ we shall set $p(i) = ((2^m + 1) \times q(2^m)) + 1$. Then $P_3(p, 1) = p(((2^m + 1) \times q(2^m)) + 1)$ and $P_2(p \times q, 2^m) = (((2^m + 1) \times q(2^m)) + 1) \times q((2^m + 1) \times q(2^m))$, so choosing $p(((2^m + 1) \times q(2^m)) + 1) = ((((2^m + 1) \times q(2^m)) + 1) \times q((2^m + 1) \times q(2^m))) + 1$ provides the contradiction for $n = 2$. We can continue in the same fashion indefinitely, thus obtaining the remaining cases.
}\end{proof}}

\section{Failure of cartesian closure}
We shall show that the category of $\reg$-represented spaces and polynomial-time computable functions is not cartesian closed. For this we define the functions $\Phi_n : \reg \to \reg$ via $\Phi_0 (\varphi) (w) = w$ and $\Phi_{n+1}(\varphi)(w) = \varphi(\Phi_{n} (\varphi) (w))$. Then computing $\Phi_n(\varphi)(w)$ takes time $\Omega(P_n(|\varphi|)(|w|))$, as already the length of the output provides a lower bound.

\begin{theorem}
\label{theo:functionspaces}
Let the second-order polynomial $P$ witness polynomial-time computability of the function $F : \subseteq \reg \times \reg \to \reg$. For no $\psi \in \reg$ we may have $F(\psi,\varphi) = \Phi_{\deg(P)+1}(\varphi)$ for all $\varphi \in \reg$.
\end{theorem}
\begin{proof}
If one considers the runtime bounds available for $F$ by assumption, and for $\Phi_{\deg(P)+1}$ as above, the claim becomes a consequence of Theorem \ref{theo:growth}.
\end{proof}

\begin{corollary}
\label{corr:functionspaces}
There cannot be an exponential in the category of $\reg$-represented spaces and polynomial-time computable functions.
\end{corollary}
\begin{proof}
Any realizer of the evaluation operation would violate Theorem \ref{theo:functionspaces}.
\end{proof}
\section{Clocked Type-Two machines}
Despite the negative result above, we can identify spaces of functions with some of the desired properties of exponentials. The required technical tool is a type-two version of clocked Turing machines. We pick a Universal Turing Machine (UTM) $M$ which simulates efficiently, meaning that on input $n, \varphi, w$ the time $M$ needs to compute the output of the $n$th Oracle Turing machine on input $w$ with oracle $\varphi$ is bounded by a quadratic polynomial in $n$ and the time $T$ needed by the $n$th Turing machine itself to compute the output on $w$ with oracle $\varphi$ (\footnote{A straight-forward adaption of the classical result by \name{Hennie} and \name{Stearns} \cite{hennie} provides the existence of such a universal machine.}). Then $M$ is extended by a clock evaluating the standard second-order polynomial\footnote{More generally, we could use an arbitrary time-constructible function in place of $P_m$. That $P_m$ actually is time-constructible is witnessed by $\Phi_m$.} $P_m$ on $|\langle n, \varphi\rangle|, |w|^l$ for fixed $m$ and some $l \in \mathbb{N}$ encoded as $(x \mapsto x^l) \in \mon$ and aborts the computation of $M$ once the runtime exceeds the value of $P_m$. Denote the resulting machine with $M^{T=P_m}$. The runtime of $M^{T=P_m}$ can be bounded by $KP_{m+1}^2 + K$ for some constant $K \in \mathbb{N}$. In particular we find that the second-order degree of the runtime of $M^{T=P_m}$ is $m + 1$.

\begin{theorem}
\label{theo:eval}
For any partial function $f : \subseteq \reg \to \reg$ computable in polynomial time $P$ with $\deg(P) \leq m$ there are some $\psi \in \reg$, $n, l \in \mathbb{N}$ such that for any $\varphi \in \dom(f)$ we find $f(\varphi) = M^{T=P_m}(\langle n, \langle \varphi, \psi\rangle, x^l\rangle)$.
\end{theorem}
\begin{proof}\omitted{
Pick some $\psi \in \reg$, $l \in \mathbb{N}$ such that $|\psi| \in \mon$, $l$ satisfy the criterion in Lemma \ref{lemma:degreeuniversal}, and some $n$ that is an index of the machine computing $f$ in time $P$. The former guarantees that the clock of $M^{T=P_m}$ does not abort the computation on valid input; its underlying universal Turing machine then works as intended.}
\end{proof}

Based on the preceding theorem, we see that rather than a single function space, we obtain a family of function spaces indexed by a natural number corresponding to the second-order degree. Given two $\reg$-represented spaces $\mathbf{X}$, $\mathbf{Y}$ we define the function space $\mathcal{C}^{T=P_m}(\mathbf{X}, \mathbf{Y})$ by letting $\langle n, \psi, x^l\rangle \in \reg$ be a name for $f : \mathbf{X} \to \mathbf{Y}$ if $\varphi \mapsto M^{T=P_m}(\langle n,\langle \varphi, \psi\rangle,x^l\rangle)$ is a realizer of $f$. This definition just enforces that $\textrm{Eval} : \mathcal{C}^{T=P_m}(\mathbf{X}, \mathbf{Y}) \times \mathbf{X} \to \mathbf{Y}$ is computable with polynomial time bound $KP_{m+1}^2 + K$.

We can then reformulate Theorem \ref{theo:functionspaces} as $\mathcal{C}^{T=P_m}(\reg, \reg) \subsetneq \mathcal{C}^{T=P_{m+1}}(\reg, \allowbreak \reg)$ and Theorem \ref{theo:eval} as $f \in \mathcal{C}^{T=P_m}(\mathbf{X}, \mathbf{Y})$ for any $f : \mathbf{X} \to \mathbf{Y}$ computable in a polynomial time-bound of $\deg \leq m$. We can easily obtain an even stronger version of the latter by adapting the proof:
\begin{corollary}
\label{corr:functionsnames}
For a function $f : \mathbf{X} \to \mathbf{Y}$ the following properties are equivalent:
\begin{enumerate}
\item $f$ is computable in polynomial time $P$ with $\deg(P) \leq m$.
\item $f \in \mathcal{C}^{T=P_m}(\mathbf{X}, \mathbf{Y})$ has a polynomial time computable name.
\end{enumerate}
\end{corollary}

As a name for a function in $\mathcal{C}^{T=P_m}(\mathbf{X}, \mathbf{Y})$ contains enough information to actually evaluate it, we immediately obtain the following connection to be computability-theoretic setting:

\begin{observation}
$\id : \mathcal{C}^{T=P_m}(\mathbf{X}, \mathbf{Y}) \to \mathcal{C}(\mathbf{X},\mathbf{Y})$ is computable.
\end{observation}

Obtaining oracles allowing function evaluation within given time-bounds is not necessarily computable, even if the existence of such bounds is promised:

\begin{proposition}
There is a $\reg$-represented space $\mathbf{X}$ and a function $H : \mathbf{X} \to \reg$ such that $H \in \mathcal{C}^{T=P_{m+1}}(\mathbf{X}, \reg)$ has a polynomial-time computable name, $H \in \mathcal{C}^{T=P_{m}}(\mathbf{X}, \reg)$, but $H \in \mathcal{C}^{T=P_{m}}(\mathbf{X}, \reg)$ has no computable name.
\begin{proof}
Let $h : \mathbb{N} \to \{0,1\}$ be the Halting problem, and let $\mathbf{X} := \{\phi \in \reg \mid |\Phi_{m+1}(\phi)(v)| \mod 2 = h(|v|)\}$. Define $H : \mathbf{X} \to \reg$ via $H(\phi)(v) = h(|v|)$. On the one hand, $H$ is clearly linear-time reducible to $\Phi_{m+1}$, and as such has a polynomial-time computable name in $\mathcal{C}^{T=P_{m+1}}(\mathbf{X}, \reg)$. On the other hand, $H$ is clearly linear-time reducible to $h$, and as such is computable in linear time relative to an oracle -- thus $H \in \mathcal{C}^{T=P_{m}}(\mathbf{X}, \reg)$.

As the only restriction for membership in $\mathbf{X}$ is given via the values of $\Phi_{m+1}$, any function computable in second-order time $P_m$, even if equipped with a computable oracle, cannot solve $H$ by Theorem \ref{theo:functionspaces}.
\end{proof}
\end{proposition}

\begin{corollary}
There is a $\reg$-represented space $\mathbf{X}$ such that the polynomial-time computable map $\id : \mathcal{C}^{T=P_m}(\mathbf{X}, \reg) \to \mathcal{C}^{T=P_{m+1}}(\mathbf{X}, \reg)$ is not computably invertible.
\end{corollary}

We shall conclude this section by noting some nice closure properties of the slicewise polynomial-time function spaces:

\begin{proposition}
\label{prop:functionspacesbasics}
The following functions are polynomial-time computable:
\begin{enumerate}
\item $\operatorname{eval} : \mathcal{C}^{T=P_m}(\mathbf{X}, \mathbf{Y}) \times \mathbf{X} \to \mathbf{Y}$ defined by $\operatorname{eval}(f, x) = f(x)$.
\item $\operatorname{curry} : \mathcal{C}^{T=P_m}(\mathbf{X} \times \mathbf{Y}, \mathbf{Z}) \to \mathcal{C}^{T=P_m}(\mathbf{X}, \mathcal{C}^{T=P_m}(\mathbf{Y}, \mathbf{Z}))$ defined by $\operatorname{curry}(f) = x \mapsto (y \mapsto f(x, y))$.
\item $\operatorname{uncurry} : \mathcal{C}^{T=P_m}(\mathbf{X}, \mathcal{C}^{T=P_n}(\mathbf{Y}, \mathbf{Z})) \to \mathcal{C}^{T=P_{m+n}}(\mathbf{X} \times \mathbf{Y}, \mathbf{Z})$ defined by $\operatorname{uncurry}(f) = (x, y) \mapsto f(x)(y)$.
\item $\circ : \mathcal{C}^{T=P_m}(\mathbf{Y}, \mathbf{Z}) \times \mathcal{C}^{T=P_n}(\mathbf{X}, \mathbf{Y}) \to \mathcal{C}^{T=P_{n+m}}(\mathbf{X}, \mathbf{Z})$, the composition of functions
\item $\times : \mathcal{C}^{T=P_m}(\mathbf{X}, \mathbf{Y}) \times \mathcal{C}^{T=P_m}(\mathbf{U}, \mathbf{Z}) \to \mathcal{C}^{T=P_m}(\mathbf{X} \times \mathbf{U}, \mathbf{Y} \times \mathbf{Z})$
\item $\operatorname{const} : \mathbf{Y} \to \mathcal{C}^{T=P_m}(\mathbf{X}, \mathbf{Y})$ defined by $\operatorname{const}(y) = (x \mapsto y)$.
\end{enumerate}
\end{proposition}
\begin{proof}
All these results are obtained by standard constructions on Turing machines (as in \cite[Proposition 3.3]{pauly-synthetic}) coupled with a straight-forward analysis of the asymptotic runtime.
\end{proof}

Instead of fixing the second-order degree of the polynomial run-times, we could consider the function space $\coprod_{n \in \mathbb{N}} \mathcal{C}^{T=P_n}$ containing all polynomial-time computable functions. The items $2.-6.$ from Proposition \ref{prop:functionspacesbasics} immediately carry over as polynomial-time computable. However, evaluation no longer is polynomial-time computable (Corollary \ref{corr:functionspaces}).

\hide{TODO: comment on evaluation on $\coprod_{n \in \mathbb{N}} \mathcal{C}^{T=P_n}$ and parameterized complexity}

\section{Effectively polynomial-bounded spaces}
\label{sec:epb}
Our next goal is to investigate restrictions we can employ on $\mathbf{X}$ (and later on $\mathbf{Y}$) in order to force the collapse of the time hierarchy $\mathcal{C}^{T=P_m}(\mathbf{X}, \mathbf{Y}) \subseteq \mathcal{C}^{T=P_{m+1}}(\mathbf{X}, \mathbf{Y})$. The collapse will only occur at the second level, as this is the minimal level where a query to the second-order input may depend on the result of another such query, which is required in order to fully utilize the function-argument depending on the input-argument.

\begin{definition}
We call $\mathbf{X}$ effectively polynomially bounded (epb)\footnote{
Note that the epb-condition acts on the domain of the representation only, it does not relate to any hypothetical additional structure available on $\mathbf{X}$ (such as a metric). In particular, this condition is unrelated to the notion of a concise representation introduced by \name{Weihrauch} in \cite{weihrauchf}.}
, iff it admits a $\reg$-representation $\delta_\mathbf{X}$ such that there is a constant $c \in \mathbb{N}$ and a monotone polynomial $Q : \mathbb{N} \to \mathbb{N}$ s.t.: \[\forall \varphi \in \dom(\delta_\mathbf{X}) \ \forall i \in \mathbb{N} \ \ |\varphi|(i) \leq c|\varphi|(c)^cQ(i)\]
\end{definition}

\begin{theorem}
\label{theo:epbcollapse}
Let $\mathbf{X}$ be epb. Then for any $m \geq 2$ we find $\mathcal{C}^{T=P_2}(\mathbf{X}, \mathbf{Y}) \cong \mathcal{C}^{T=P_m}(\mathbf{X}, \mathbf{Y})$ where $\cong$ denotes polytime isomorphic.
\end{theorem}
\begin{proof}
It suffices to show only the direction $\subseteq : \mathcal{C}^{T=P_{m}}(\mathbf{X}, \mathbf{Y}) \to \mathcal{C}^{T=P_2}(\mathbf{X}, \mathbf{Y})$. Let $M$ be the UTM used in the definition of $\mathcal{C}^{T=P_{m}}(\mathbf{X}, \mathbf{Y})$, let $M'$ behave with the oracle $\langle \varphi, \langle \psi, \psi'\rangle\rangle$ in exactly the same way as $M$ does with $\langle \varphi, \psi\rangle$, and then finally, use $M'$ to define $\mathcal{C}^{T=P_2}(\mathbf{X}, \mathbf{Y})$.

The assumption that $\mathbf{X}$ is epb allows us to estimate:
$$\begin{array}{rcl} P_m(|\langle \varphi, \psi\rangle|)(k) & = & |\langle \varphi, \psi\rangle|(P_{m-1}(|\langle \varphi, \psi\rangle|) (k))\\
& \leq & c|\varphi|(c)^cQ(P_{m-1}(|\langle \varphi, \psi\rangle|)( k)) \times |\psi|(P_{m-1}(|\langle \varphi, \psi\rangle|)( k)) \\
& \leq & (cQ^c \times |\psi|)(P_{m-1}(|\langle \varphi, \psi\rangle|)( (k+1)^c)) \\
& \leq & (cQ^c \times |\psi|)\left ((cQ^c + |\psi|)(P_{m-2}(|\langle \varphi, \psi\rangle|)( (k+1)^{c^2}))\right ) \\
& \leq & (cQ^c \times |\psi|)^{(m)}(|\langle \varphi, \psi\rangle|((k+1)^{c^m})) \\
& \leq & P_2(|\langle \varphi, \psi\rangle| \times (cQ^c \times |\psi|)^{(m)})( (k+1)^{c^m})\end{array}$$
Now given $\psi$, we can compute some $\psi'$ with $|\langle \varphi, \psi\rangle| \times (cQ^c \times |\psi|)^{(m)} \leq |\langle \langle \varphi, \psi\rangle, \psi'\rangle|$ in polynomial time (note that $Q$, $c$ and $m$ are all constants here). The $l$ in the original name is replaced by $lc^m$.
\end{proof}

It is worthwhile pointing out that the function spaces for computability do not only contain the computable functions as elements, but comprise exactly the continuous functions as discussed very well in \cite{bauer}, yielding a structure dubbed \emph{category extension} in \cite{pauly-synthetic,paulysearchproblems}. This is due to the fact that the (partial) functions $f : \subseteq \Baire \to \Baire$ arising as sections of computable (partial) functions $F : \subseteq \Baire \times \Baire \to \Baire$ are just the continuous functions.

In a similar way, we shall investigate which functions appear in a space $\mathcal{C}^{T=P_2}(\mathbf{X}, \mathbf{Y})$ for epb $\mathbf{X}$. It turns out that (a modification of) uniform continuity plays a central role. A connection between run-time bounds and the modulus of continuity was also found for multivalued functions in \cite{paulyziegler}.

\begin{definition}
We call a partial function $f : \subseteq \reg \to \reg$ \emph{polytime-locally uniformly continuous}, if there is a polynomial-time computable function $\chi : \subseteq \reg \to \mathbb{N}$, such that $\dom(f) \subseteq \dom(\chi)$ and any $f|_{\chi^{-1}(\{n\})}$ is uniformly continuous.
\end{definition}

\begin{theorem}
Let $\mathbf{X} \subseteq \reg$ be epb. Then for $f : \mathbf{X} \to \reg$ the following are equivalent:
\begin{enumerate}
\item $f$ is polytime-locally uniformly continuous
\item $f \in \mathcal{C}^{T=P_2}(\mathbf{X}, \reg)$
\end{enumerate}
\end{theorem}
\begin{proof}
\begin{description}
\item[$1. \Rightarrow 2.$] Given Theorem \ref{theo:epbcollapse} and Corollary \ref{corr:functionsnames}, it suffices to show that such an $f$ is polynomial-time computable relative to some oracle $\psi$. We start by some $\Lambda \in \mon$ such that $i \mapsto \Lambda(\langle n, i\rangle)$ is a modulus of continuity of $f|_{\chi^{-1}(\{n\})}$. Then $f(\varphi)(u)$ depends only on values $\varphi(w)$ with $|w| \leq \Lambda(\langle \chi(\varphi), |u|\rangle)$, and we may encode this dependency in some table $\psi$. In order to write the query to $\psi$, the machine needs time $2^{\Lambda(\langle \chi(\varphi), |u|\rangle)}$. By providing $\langle2^\Lambda, \psi\rangle$ as an oracle, this time is made available.
\item[$2. \Rightarrow 1.$] \omitted{By continuing the estimate from the proof in Theorem \ref{theo:epbcollapse} we obtain an upper bound for the evaluation of $f$ given its $\mathcal{C}^{T=P_2}(\mathbf{X}, \reg)$-name $\psi$ depending only on $\psi$, $l$ and $|\varphi(c)|$, but beyond that not on $\varphi$. In particular, for fixed $|\varphi(c)|$, there is a bound $\lambda : \mathbb{N} \to \mathbb{N}$ such that to compute $f(\varphi)(w)$, $\varphi$ is only queried on inputs $v$ with $|v| \leq \lambda(|w|)$ -- but this is uniform continuity. It is clear that $\varphi \mapsto |\varphi(c)|$ is a polynomial-time computable map.}
\end{description}
\end{proof}

Note that the same argument used for $1. \Rightarrow 2.$ in the preceding proof also establishes that $\mathcal{C}^{T=P_2}(\mathbb{R}, \mathbb{R})$ contains all the continuous functions, where $\mathbb{R}$ is represented as suggested in \cite{kawamura}, as observed by the first author in \cite{kawamura6}. In particular, $\mathbb{R}$ as defined there is an epb space -- and the best example of an epb space available to us.

\begin{observation}
If $\mathbf{X}$ and $\mathbf{Y}$ are epb, then so are $\mathbf{X} + \mathbf{Y}$ and $\mathbf{X} \times \mathbf{Y}$. Any subspace of an epb-space is epb itself. However, $\mathcal{C}^{T=P_2}(\mathbf{X}, \mathbf{Y})$ is not necessarily epb. If $\mathbf{X} \cong \mathbf{X}'$, we also cannot conclude that $\mathbf{X}'$ is epb, as $\mathbf{X}'$ may have superfluous fast-growing names\footnote{This aspect raises the question whether there is a convenient characterization of representations that are polynomial-time equivalent to an epb representation.}.
\end{observation}

\section{Padding and polytime admissibility}
In this section we shall explore two distinct but similar arguments based on using padding-like concepts on the codomain of a function in order to make time bounds irrelevant. This technique both reveals polynomial-time admissibility as a far too restrictive concept (as opposed to computable admissibility) and allows us to draw some conclusions about degree structures.

We define a $\reg$-representation $\pi$ of Cantor space via $\dom(\pi) = \{\varphi \in \reg \mid \operatorname{range}(|\varphi|) = \mathbb{N}\}$ and $\pi(\varphi)(i) = \varphi(0^n)(i)$ where $n = \min \{j \in \mathbb{N} \mid |\varphi(0^j)| = i\}$. Now any Cantor-representation $\delta$ can be turned into a $\reg$-representation by composing with $\pi$, and by this we obtain a strong correspondence between computability and polynomial-time computability.
\begin{proposition}
\label{prop:padding}
A function $f : \mathbf{X} \to (Y, \delta_\mathbf{Y})$ is computable if and only if $f : \mathbf{X} \to (Y, \delta_\mathbf{Y} \circ \pi)$ is polynomial-time computable.
\end{proposition}
\begin{proof}\omitted{
The map $\pi$ is computable, this provides one direction. For the other direction, note that a computation providing a result in $(Y, \delta_\mathbf{Y} \circ \pi)$ can safely be delayed as long as required to stay within any given time bound.}
\end{proof}

Weihrauch reducibility (e.g.~\cite{brattka2,brattka3,paulybrattka,paulykojiro}) is a computable many-one reduction between multivalued functions that serves as the basis of a metamathematical research programme. Likewise, a reduction that could be called polynomial-time Weihrauch reducibility has been investigated by some authors (e.g.~\cite{beame,kawamura}). In \cite{paulysearchproblems} abstract principles were demonstrated that provide a very similar degree structure for both. Let $(\mathfrak{W}, \oplus, +, \times)$ and $(\mathfrak{P}, \oplus, +, \times)$ be the corresponding degree structures for Weihrauch reducibility and polynomial-time Weihrauch reducibility. We then find:
\begin{corollary}
$(\mathfrak{W}, \oplus, +, \times)$ embeds as a substructure into $(\mathfrak{P}, \oplus, +, \times)$.
\end{corollary}

The characterization of admissibility that admits a translation into the setting of computational complexity is due to \name{Schr\"oder} \cite{schroder5} (see also \cite{pauly-synthetic}). Given the \Sierp space $\mathbb{S}$ and the function space $\mathcal{C}(-,-)$, we find that there is a canonic map $\kappa_\mathbf{X} : \mathbf{X} \to \mathcal{C}(\mathcal{C}(\mathbf{X}, \mathbb{S}),\mathbb{S})$ with $\kappa(x)(f) = f(x)$. A space $\mathbf{X}$ is called computably admissible, if $\kappa_\mathbf{X}$ admits a computable partial inverse.

The space $\mathbb{S}$ has the underlying set $\{\top, \bot\}$, and the representation $\delta_\mathbb{S} : \reg \to \mathbb{S}$ defined by $\delta_\mathbb{S}(\varphi) = \top$ iff $\exists w \ . \ \varphi(w) = 1$. By the same argument as Proposition \ref{prop:padding}, any computable function into $\mathbb{S}$ is computable in polynomial time -- in fact, even linear time suffices. Thus, just as in Section \ref{sec:epb} we can use the space $\mathcal{C}^{T=P_1}(\mathbf{X}, \mathbb{S})$ as a function space and subsequently obtain a definition of polynomial-time admissibility by calling $\mathbf{X}$ polynomial-time admissible iff the (polynomial-time computable) map $\kappa_\mathbf{X} : \mathbf{X} \to \mathcal{C}^{T=P_1}(\mathcal{C}^{T=P_1}(\mathbf{X}, \mathbb{S}), \mathbb{S})$ has a polynomial-time computable partial inverse. However, this notion is of limited use:
\begin{proposition}
If $x \in \mathbf{X}$ for polynomial-time admissible $\mathbf{X}$ has a computable name, then it has a polynomial-time computable name.
\end{proposition}
\begin{proof}\omitted{
As polynomial-time computable functions preserve polynomial-time computable names, this follows from a function $f : \mathcal{C}^{T=P_1}(\mathbf{X}, \mathbb{S}) \to \mathbb{S}$ being polynomial-time computable iff it is computable together with Corollary \ref{corr:functionsnames}.}
\end{proof}

Note that this implies that all the representations suggested in \cite{kawamura} fail to be polynomial-time admissible, despite appearing to be very reasonable choices\footnote{Nevertheless, there are non-trivial polynomial-time admissible spaces. In particular, any space $\mathcal{C}^{T=P_1}(\mathbf{X}, \mathbb{S})$ will be polynomial-time admissible. Consequently, we find that there is a polynomial-time admissible space in any equivalence class regarding computable translations that is computably admissible -- but for these spaces, the formally defined polynomial-time computability actually is just computability, without any complexity-theoretic flavour to it.}.

\section{Conclusions}
The trusted techniques developed for the theory of represented spaces and computable functions are insufficient to fully comprehend polynomial-time computability. Function spaces are not always available, and even where they are, they might differ from the familiar one of the continuous functions\footnote{This observation was also made by \name{F\'er\'ee} and {Hoyrup} \cite{feree} (see also \cite{feree2}), and they suggested to use higher-order functionals on the machine level to retain spaces of continuous function with efficient evaluation. However, as shown by \name{Schr\"oder} (personal communication), this would change the notion of computability, too.}. Instead, some form of uniform continuity will be appear as the central notion.

What can be used as a guiding principle for the choice of representations is the epb property. If compatible with other criteria, choosing a representation that makes a space epb also makes function spaces well-behaved. For example, separable metric spaces are traditionally represented by encoding points by fast converging sequences of basic elements. For computability theory it does not matter what \emph{fast} means -- for complexity theory it does. A sensible choice could be: As fast as possible while retaining the epb property. Whether this already determines a representation up to polynomial-time equivalence is open, though.

\bibliographystyle{eptcs}
\bibliography{funcsp}

\section*{Acknowledgements}
The second author is grateful to Anuj Dawar, Carsten R\"osnick and Martin Ziegler for valuable discussions pertaining to the topic of the paper.
This work is supported in part by the Japanese
Grant-in-Aid for Scientific Research (\mbox{Kakenhi}) 24106002, and the Marie Curie International Research Staff Exchange Scheme \emph{Computable
Analysis}, PIRSES-GA-2011- 294962.

\end{document}